\documentclass[a4paper, 11pt]{article}

\usepackage{epsfig}
\usepackage{amsmath}
\usepackage{amssymb}
\usepackage{latexsym}
\usepackage{amsfonts}
\usepackage{amsthm}
\usepackage{color}
\usepackage[numbers,sort&compress]{natbib}
\usepackage{graphicx}
\ifx\pdfoutput\undefined  \else
 \fi
\usepackage[centerlast]{caption2}

\newtheorem{lemma}{Lemma}
\newtheorem{theorem}{Theorem}

\numberwithin{figure}{section} \numberwithin{definition}{section}
\numberwithin{observation}{section} \numberwithin{lemma}{section}\numberwithin{corollary}{section}
\numberwithin{theorem}{section} \numberwithin{proposition}{section}
\numberwithin{conjecture}{section} \numberwithin{table}{section}

\begin{document}
\title{
{The crossing numbers of $K_m\times P_n$ and $K_m\times
C_n$}\footnote{The research is supported by NSFC
(No.60973014,\ 61170303)}}
\author{
Yuansheng Yang \footnote {corresponding
author's email : yangys@dlut.edu.cn}, \ Baigong Zheng, \ Xirong Xu, \ Xiaohui Lin \\
School of Computer Science and Technology\\
Dalian University of Technology, Dalian, 116024, P.R. China }

\date{}
\maketitle
\begin{abstract}
The {\it crossing number} of a graph $G$ is the minimum number of
pairwise intersections of edges in a drawing of $G$. In this paper,
we study the crossing numbers of $K_{m}\times P_n$ and $K_{m}\times
C_n$.
\bigskip

\end{abstract}

\section{Introduction and preliminaries}

\indent \indent Let $G$ be a graph, $V(G)$ the vertex set and $E(G)$ the edge set. The crossing number of $G$, denoted by $cr(G)$, is the smallest number of pairwise crossings of edges among all
drawings of $G$ in the plane. We use $D$ to denote a drawing of a graph $G$ and $\nu(D)$ the number of crossings in $D$. It is clear that $cr(G)\leq \nu(D)$.

Let $E_1$ and $E_2$ be two disjoint subsets of an edge set $E$. The number of the crossings formed by an edge in $E_1$ and another edge in $E_2$ is denoted by $\nu_D(E_1,E_2)$ in a drawing $D$.
The number of the crossings that involve a pair of edges in $E_1$ is denoted by $\nu_D(E_1)$. Then $\nu(D)=\nu_D(E)$. By counting the numbers of crossings in $E$, we have
\begin{lemma}\label{lemma: E12}
$\nu_D(E_1\cup E_2)=\nu_D(E_1)+\nu_D(E_2)+\nu_D(E_1,E_2)$.
\end{lemma}

The $Kronecker$ $product$ $G\times H$ of graphs $G$ and $H$ has
vertex set $V(G\times H)=V(G)\times V(H)$ and edge set $E(G\times
H)=\{\{(a,x),(b,y)\}:\{a,b\}\in E(G)$ and $\{x,y\}\in E(H)\}$. (The
product is also known as direct product, cardinal product, cross
product and graph conjunction.)

Computing the crossing number of graphs is a complicated yet
classical problem. And it is proved that the problem is NP-complete
by Garey and Johnson \cite{GJ83}.

In literature, the Cartesian product has been paid more attention\cite{A04,Z07,Z08,Z081}, while Kronecker product has fewer results on the crossing number\cite{JD12}.

In this paper, we study the crossing numbers of the Kronecker
product $K_m\times P_n$ and $K_m\times C_n$. In Section 2, we give
an upper bound of $cr(K_{m}\times P_n)$ for $n\geq 4$ and $m\geq 4$.
In Section 3, we give an upper bound of $cr(K_{m}\times C_n)$ for
$n\geq 3$ and $m\geq 4$. In Section 4, we give lower bounds of
$cr(K_{m}\times P_n)$ and $cr(K_{m}\times C_n)$.

\section{Upper bound of $cr(K_{m}\times P_n)$}

\indent \indent Let
$$\begin{array}{llll}
V(K_{m}\times P_n)=&\{(i,j) \ |\ 0\leq i\leq m-1 \mbox{ and }  0\leq j\leq n-1\},\\
E(K_{m}\times P_n)=&\{((i_1,j),(i_2,j+1))\ |\ 0\leq i_1\neq i_2\leq
m-1 \mbox{ and } 0\leq j\leq n-2 \},
\end{array}$$
where the first subscript is modulo $m$.

For $0\leq j\leq n-2$, let
$$\begin{array}{llll}
E^j=&\{((i_1,j),(i_2,j+1))\ |\ 0\leq i_1\neq i_2\leq m-1\}.
\end{array}$$
Then
$$\begin{array}{llll}
\bigcup_{j=0}^{n-2}E^j=E(K_{m}\times P_n),\ \ E^{j_1}\bigcap
E^{j_2}=\emptyset(0\leq j_1\neq j_2\leq n-2).
\end{array}$$

In Figure  \ref{fig: Dp14}, we exhibit drawings $D_{m,4}$ of $K_m\times P_4$ in a cylinder for $m\leq 10$. A cylinder can be `assembled' from a polygon by identifying one pair of opposite sides
of a rectangle\cite{BW78}.

\begin{figure}
\centering
\includegraphics[scale=1.0]{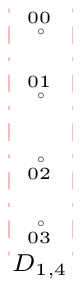}
\includegraphics[scale=1.0]{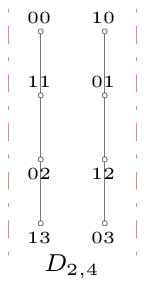}
\includegraphics[scale=1.0]{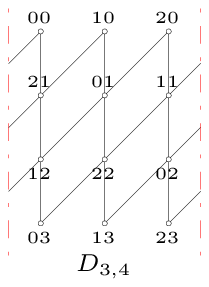}
\includegraphics[scale=1.0]{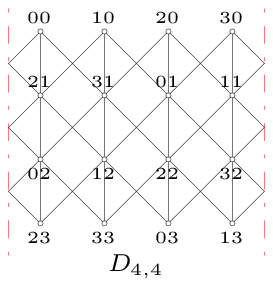}
\includegraphics[scale=1.0]{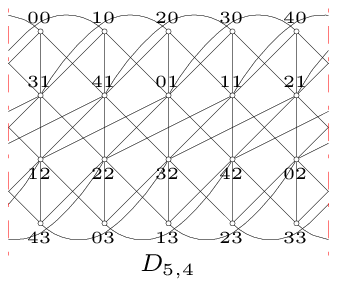}
\includegraphics[scale=1.0]{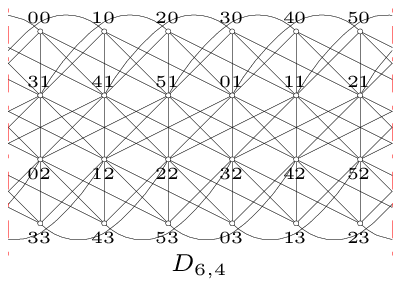}
\includegraphics[scale=1.0]{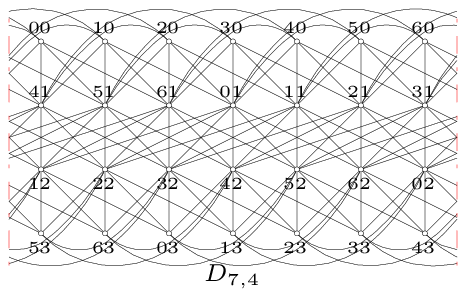}
\hspace{60bp}
\includegraphics[scale=1.0]{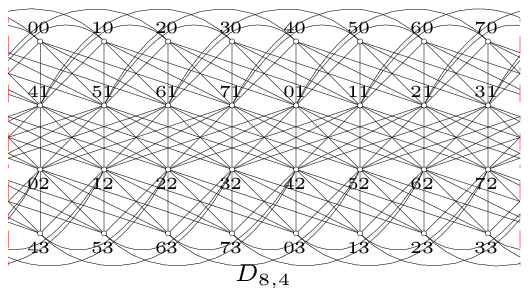}
\includegraphics[scale=1.0]{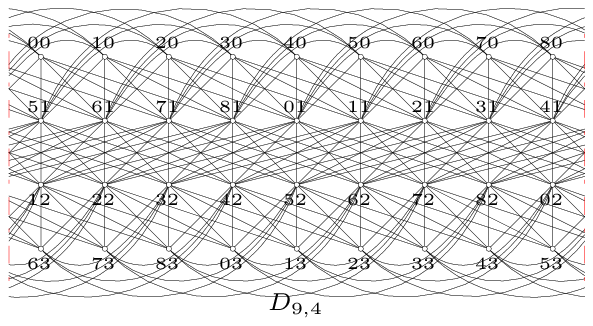}
\hspace{24bp}
\includegraphics[scale=1.0]{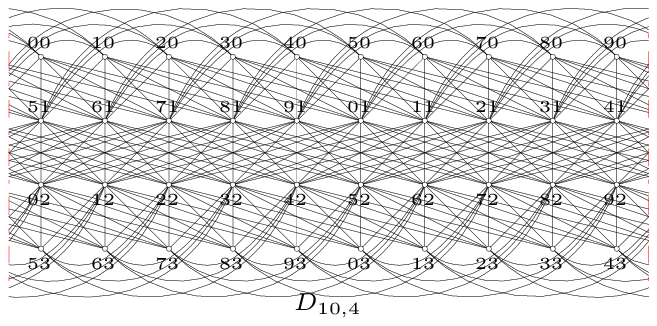}
\caption{\small{Drawings $D_{m,4}$ for $m\leq 10$}}\label{fig: Dp14}
\end{figure}

By counting the numbers of crossings in $D_{m,n}$, we have
\begin{lemma}\label{lemma: Dp4}
For $n\geq 4$,
$$\begin{array}{llll}
\nu(D_{m,n})&=\left \{ \begin{array}{llll}
\frac{(n-1)m(m-1)(m-2)(m-3)}{6}-\frac{m(m-3)(m^2+6m-31)}{24} & \mbox{ for odd }m\geq 5,\\
\frac{(n-1)m(m-1)(m-2)(m-3)}{6}-\frac{m(m-4)(m^2+10m-48)}{24} & \mbox{ for even }m\geq 4.
                \end{array}
     \right .
\end{array}$$
\end{lemma}
\begin{proof} Since $\nu_{D_{m,n}}(E^{j_1},E^{j_2})=0$ for $0\leq j_1\neq j_2\leq n-2$, by Lemma \ref{lemma: E12}, we have
$\nu(D_{m,n})=\sum_{j=0}^{n-2}\nu_{D_{m,n}}(E^{j})=2\nu_{D_{m,n}}(E^{0})+(n-3)\nu_{D_{m,n}}(E^{1})$.
For $m\geq 4$,
$$\begin{array}{llll}
\nu_{D_{m,n}}(E^1)=m\sum_{j=0}^{m-3}\sum_{i=0}^{j}i=\frac{m(m-1)(m-2)(m-3)}{6}.\hspace{400bp}
\end{array}$$
For odd  $m\geq 5$,
$$\begin{array}{llll}
\nu_{D_{m,n}}(E^0)=\nu_{D_{m,n}}(E^1)-m\sum_{j=2}^{\frac{m-1}{2}}(\sum_{i=2}^{j}i-1)=\frac{m(m-1)(m-2)(m-3)}{6}-\frac{m(m-3)(m^2+6m-31)}{48}.\hspace{400bp}
\end{array}$$
For even $m\geq 4$,
$$\begin{array}{llll}
\nu_{D_{m,n}}(E^0)=\nu_{D_{m,n}}(E^1)-m\sum_{j=3}^{\frac{m}{2}}(\sum_{i=3}^{j}i-1)=\frac{m(m-1)(m-2)(m-3)}{6}-\frac{m(m-4)(m^2+10m-48)}{48}.\hspace{400bp}
\end{array}$$
Hence,
$$\begin{array}{llll}
\nu(D_{m,n})&=2\nu_{D_{m,n}}(E^{0})+(n-3)\nu_{D_{m,n}}(E^{1})\\
&=\left \{ \begin{array}{llll}
\frac{(n-1)m(m-1)(m-2)(m-3)}{6}-\frac{m(m-3)(m^2+6m-31)}{24} & \mbox{ for odd }m\geq 5,\\
\frac{(n-1)m(m-1)(m-2)(m-3)}{6}-\frac{m(m-4)(m^2+10m-48)}{24} & \mbox{ for even }m\geq 4.
                \end{array}
     \right .
\end{array}$$
\end{proof}
It is easy to verify that $cr(K_m\times P_n)=0$ for $m=1,2,3.$ (See Figure \ref{fig: Dp14}). For $m\geq 4$, by Lemma \ref{lemma: Dp4} we have
\begin{theorem}\label{theorem: kmp4}
For $n\geq 4$,
$$\begin{array}{llll}
cr(K_m\times P_n)&\leq\left \{ \begin{array}{llll}
\frac{(n-1)m(m-1)(m-2)(m-3)}{6}-\frac{m(m-3)(m^2+6m-31)}{24} & \mbox{ for odd }m\geq 5,\\
\frac{(n-1)m(m-1)(m-2)(m-3)}{6}-\frac{m(m-4)(m^2+10m-48)}{24} & \mbox{ for even }m\geq 4.
                \end{array}
     \right .
\end{array}$$

\end{theorem}
We will discuss the crossing number of $K_m\times P_n$ for $n=2,3$ in another paper.

\section{Upper bound of $cr(K_{m}\times C_n)$}
\indent \indent It is easy to verify that $cr(K_m\times C_n)=0$ for
$m=1,2.$ (See Figure \ref{fig: D12}). For $m=3$, $K_3\times C_n\cong
C_3\times C_n$. By \cite{JD12}, $cr(K_3\times C_3)=3$, $cr(K_3\times
C_n)\leq 6n-18$ for $3<n<9$ and $cr(K_3\times C_n)\leq 3n$ for
$n\geq 9$. In this paper, we only consider the case for $m\geq 4$.
\begin{figure}[hb]
\centering
\includegraphics[scale=1.0]{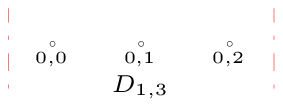}
\includegraphics[scale=1.0]{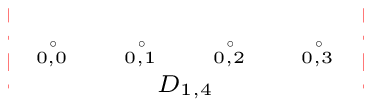}
\includegraphics[scale=1.0]{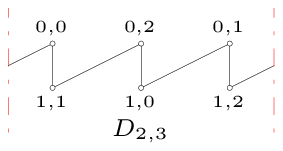}
\includegraphics[scale=1.0]{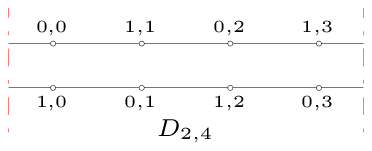}
\caption{\small{Drawings $D_{m,n}$ for $(m,n)\in
\{(1,3),(1,4),(2,3),(2,4)\}$}}\label{fig: D12}
\end{figure}

Let
$$\begin{array}{llll}
V(K_{m}\times C_n)=&\{(i,j) \ |\ 0\leq i\leq m-1, 0\leq j\leq n-1\},\\
E(K_{m}\times C_n)=&\{((i_1,j),(i_2,j+1))\ |\ 0\leq i_1\neq i_2\leq m-1, 0\leq j\leq n-1\},\\
\end{array}$$
where the first subscript is modulo $m$ and the second subscript is
modulo $n$.

For $0\leq j\leq n-1$, let
$$\begin{array}{llll}
V^j=&\{(i,j)\ |\ 0\leq i\leq m-1\},\\
E^j=&\{((i_1,j),(i_2,j+1))\ |\ 0\leq i_1\neq i_2\leq m-1\}.
\end{array}$$
Then
$$\begin{array}{llll}
\bigcup_{j=0}^{n-1}E^j=E(K_{m}\times C_n),\ \ E^{j_1}\bigcap E^{j_2}=\emptyset(0\leq j_1\neq j_2\leq n-1).
\end{array}$$

In Figure  \ref{fig: D44}, we exhibit drawings $D_{m,n}$ of $K_m\times C_n$ in a cylinder for $(m,n)\in \{(4,6),(5,6),$ $(4,7),(5,7),(6,7),(7,7)\}$. In Figure  \ref{fig: D43} and \ref{fig: D65},
we exhibit drawings $D_{m,3}$ for $4\leq m\leq 7$ and $D_{m,5}$ for $6\leq m\leq 11$ respectively.

In drawings $D_{m,n}$, vertices $(i_{0,0},0),(i_{1,0},1),\cdots, (i_{n-1,0},n-1)$ ($(i_{0,m-1},0),(i_{1,m-1},1),\cdots,$ $(i_{n-1,m-1},n-1)$) are placed equidistantly on the perimeter of the top
(bottom) disk, vertices $(i_{0,j},0),(i_{1,j},1),\cdots,(i_{n-1,j},$ $n-1)$ are placed equidistantly on the cylinder from top to down for $j$ from 1 to $m-2$, edges of $E^j$ are drawn on by
shortest helical curves on the cylinder.
\begin{figure}
\centering
\includegraphics[scale=1.0]{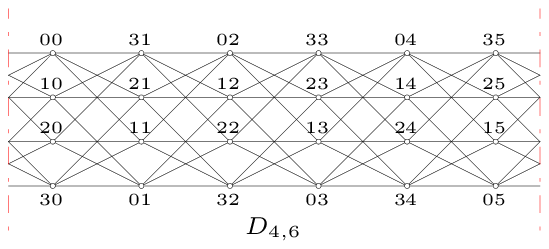}
\includegraphics[scale=1.0]{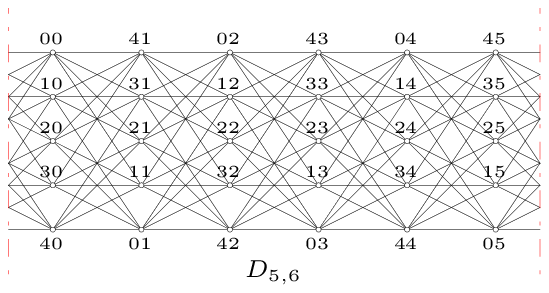}
\includegraphics[scale=1.0]{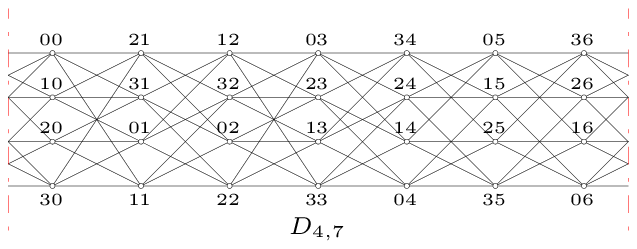}
\includegraphics[scale=1.0]{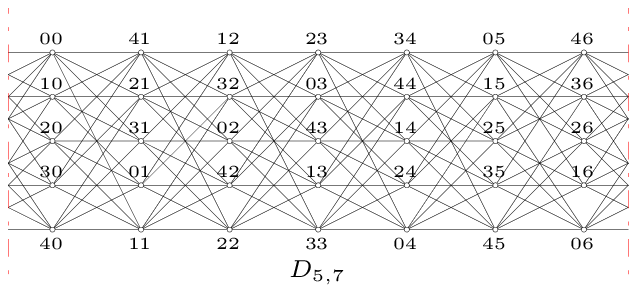}
\includegraphics[scale=1.0]{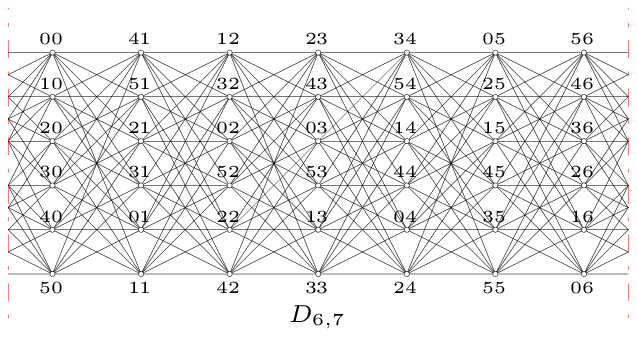}
\includegraphics[scale=1.0]{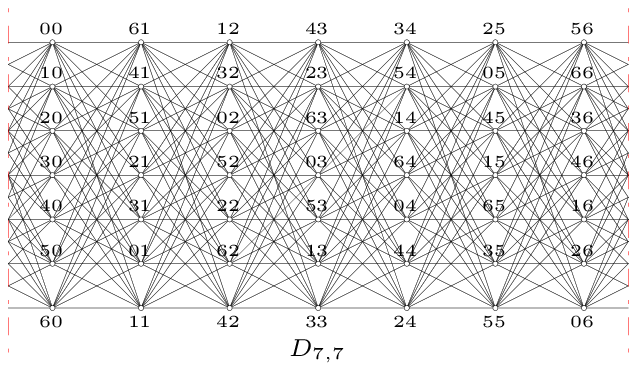}
\caption{\small{Drawings $D_{m,n}$ for $(m,n)\in
\{(4,6),(5,6),(4,7),(5,7),(6,7),(7,7)\}$}}\label{fig: D44}
\end{figure}
For $0\leq j\leq n-1$, let $f^j=(f^j(0),f^{j}(1),\cdots,f^j(m-1))$
be an arrangement of $\{0,1,\cdots,m-1\}$ such that
$i_{j-1,t}=i_{j,f^j(t)}$ for all $0\leq t\leq m-1$, where $j$ is
modulo $m$. In drawings $D_{m,n}$, $i_{0,t}=t$ for $0\leq t\leq
m-1$.

For $m\geq 4$, let
$$\begin{array}{llll}
    f_1(t)=m-1-t,\ 0\leq t\leq m-1,\\
    f_2(0)=m-1, \\
    f_2(t)=m-1-t+(-1)^{t},\ 1\leq t\leq m-2, \\
    f_2(m-1)=\frac{1-(-1)^{m}}{2},\\
    f_3(t)=m-1-t-(-1)^t,\ 0\leq t\leq m-1,\\
\end{array}$$

If $m\geq 4$ and even $n\geq 4$, $f^j=f_1$ for $0\leq j \leq n-1$.\\
\indent If $3\leq$ odd $m \leq$ odd $n$, $f^j=f_2$ for $0\leq j \leq m-1$,  $f^j=f_1$ for $m\leq j \leq n-1$.\\
\indent If $4\leq$ even $m \leq$ odd $n-1$, $f^j=f_{3-j \mbox{ mod
}2} $ for $0\leq j \leq m-1$,  $f^j=f_1$ for $m\leq j \leq n-1$.

For integer $l$, let $inv(f_l)$ be the inversion number in $f_l$. By
counting the number of crossings in $D_{m,n}$, we have
\begin{lemma}\label{lemma: D0}
If $f^j=f_l$, then
$\nu_{D_{m,n}}(E^j)=(^{m}_2)(^{m}_2)-(\sum_{t=0}^{m-1}((m-1-t)f_l(t)+t(m-1-f_l(t))-inv(f_l))$.
\end{lemma}

By Lemma \ref{lemma: D0}, we can get Lemmas \ref{lemma:
D1}-\ref{lemma: D3}:
\begin{lemma}\label{lemma: D1}
If $f^j=f_1$, then
$\nu_{D_{m,n}}(E^j)=\frac{m(m-1)(m-2)(3m-5)}{12}$.
\end{lemma}
\begin{proof}
$$\begin{array}{llll}
\nu_{D_{m,n}}(E^j)&=(^{m}_2)(^{m}_2)-(\sum_{t=0}^{m-1}((m-1-t)(m-1-t)-t(m-1-(m-1-t)))-(^{m}_2))\\
&=\frac{m(m-1)(m-2)(3m-5)}{12}.
\end{array}$$
\end{proof}
\begin{lemma}\label{lemma: D21}
If $f^j=f_2$ and $m$ is odd, then
$\nu_{D_{m,n}}(E^j)=\frac{m(m-1)(m-2)(3m-5)}{12}+\frac{m-1}{2}$.
\end{lemma}
\begin{proof}
$$\begin{array}{llll}
\nu_{D_{m,n}}(E^j)&=(^{m}_2)(^{m}_2)-((m-1)^2+2\sum_{t=1}^{\frac{m-1}{2}}((m-1-2t)(m-2t)+2t(2t-1))\\
&\ \ \ \ -(m-1+2\sum_{t=1}^{\frac{m-3}{2}}2t))\\
&=\frac{m(m-1)(m-2)(3m-5)}{12}+\frac{m-1}{2}.
\end{array}$$
\end{proof}
\begin{lemma}\label{lemma: D22}
If $f^j=f_2$ and $m$ is even, then
$\nu_{D_{m,n}}(E^j)=\frac{m(m-1)(m-2)(3m-5)}{12}+\frac{m-2}{2}$.
\end{lemma}
\begin{proof}
$$\begin{array}{llll}
\nu_{D_{m,n}}(E^j)&=(^{m}_2)(^{m}_2)-(2(m-1)^2+2\sum_{t=1}^{\frac{m-2}{2}}((m-1-2t)(m-2t)+2t(2t-1))\\
&\ \ \ \ -(m-1+2\sum_{t=1}^{\frac{m-2}{2}}(2t-1)))\\
&=\frac{m(m-1)(m-2)(3m-5)}{12}+\frac{m-2}{2}.
\end{array}$$
\end{proof}

\begin{lemma}\label{lemma: D3}
If $f^j=f_3$ and $m$ is even, then
$\nu_{D_{m,n}}(E^j)=\frac{m(m-1)(m-2)(3m-5)}{12}+\frac{m}{2}$.
\end{lemma}
\begin{proof}
$$\begin{array}{llll}
\nu_{D_{m,n}}(E^j)&=(^{m}_2)(^{m}_2)-(2\sum_{t=1}^{\frac{m}{2}}((m-2t)(m-(2t-1))+(2t-1)(2t-2))-2\sum_{t=1}^{\frac{m-2}{2}}(2t))\\
&=\frac{m(m-1)(m-2)(3m-5)}{12}+\frac{m}{2}.
\end{array}$$
\end{proof}

By Lemma \ref{lemma: E12} and Lemmas \ref{lemma: D1}-\ref{lemma:
D3}, we have
\begin{lemma}\label{lemma: D4}
For $m\geq 4$ and even $n\geq 4$, $\nu(D_{m,n})=\frac{n\cdot m(m-1)(m-2)(3m-5)}{12}$.
\end{lemma}
\begin{lemma}\label{lemma: D55}
For $4\leq m\leq$ odd $n$, $\nu(D_{m,n})=\frac{n\cdot m(m-1)(m-2)(3m-5)}{12}+\frac{m(m-1)}{2}$.
\end{lemma}

Now we consider the case of $m>$ odd $n\geq 3$. Let $r=\lfloor\frac{m}{n}\rfloor$, $s=m$ mod $n$, $s_0=\frac{n-1}{2}$ and $s_1=\frac{s}{2}$. Let {\tiny
$$\begin{array}{llll}
        f_4(d\cdot s_0)=m-d\cdot s_0-2+\frac{1-(-1)^{s_0}}{2},\ 0\leq d\leq r-1,\hspace{400bp}\\
        f_4(t+d\cdot s_0)=m-d\cdot s_0-t-1-(-1)^{t+s_0},\ n\geq 5, \ 0\leq t\leq s_0-1,\ 0\leq d\leq r-1,\\

        f_4(s_0\cdot r)=s_0\cdot r+s+r-2+\frac{1-(-1)^{s_1}}{2}, \ s\geq 2,\\
        f_4(t)=m-t-1-(-1)^{t+s_1+s_0*r},\ s\geq 4, \ s_0\cdot r+1\leq t\leq s_0\cdot r+s_1-1, \\
        f_4(s_0\cdot r+s_1+d)=(d+1)s_0-1,\ 0\leq d\leq r-1,\\
        f_4(t)=m-t-1-(-1)^{t+s_1+r+s_0*r},\ \mbox{ even }s\geq 4, \ s_0\cdot r+s_1+r\leq t\leq s_0\cdot r+s+r-2, \\
        f_4(s_0\cdot r+s+r-1)=s_0\cdot r+1-\frac{1-(-1)^{s_1}}{2},\ \mbox{ even }s\geq 2, \\
        f_4(s_0\cdot r+r)=s_0\cdot r+r,\ s=1, \\
        f_4(s_0\cdot r+s_1+r)=s_0\cdot r+s_1-1,\ \mbox{ odd } s\geq 3, \\
        f_4(s_0\cdot r+s_1+r+1)=s_0\cdot r+s_1+r,\ \mbox{ odd } s\geq 3, \\
        f_4(t)=m-t-1-(-1)^{t+s_1+r+s_0*r},\ \mbox{ odd }s\geq 7, \ s_0\cdot r+s_1+r+2\leq t\leq s_0\cdot r+s+r-2, \\
        f_4(s_0\cdot r+s+r-1)=s_0\cdot r+s_1+\frac{1-(-1)^{s_1}}{2},\ \mbox{ odd } s\geq 5, \\
        f_4(m-(d+1)s_0)=s_0\cdot r+s_1+d,\ 0\leq d\leq r-1, \\
        f_4(t)=m-t-1+(-1)^{t+r+s_0+s_0\cdot d},\ n\geq 5,  \ m-d\cdot s_0-s_0+1\leq t\leq m-d\cdot s_0-2,\ 0\leq d\leq r-1, \\
        f_4(m-d\cdot s_0-1)=d\cdot s_0+\frac{1-(-1)^{s_0}}{2},\ n\geq 5, \ 0\leq d\leq r-1.
\end{array}$$
$$\begin{array}{llll}
        f_5(t)=f_4(t),\ 0\leq t\leq s_0\cdot r-1 \mbox{ or } m-s_0\cdot r\leq t\leq m-1,\hspace{400bp}\\
        f_5(s_0\cdot r)=s_0\cdot r+s+r-1-\frac{1-(-1)^{s_1}}{2},\ s\geq 4, \\
        f_5(t)=m-t-1+(-1)^{s_0\cdot r+s_1+t},\ s\geq 6, \ s_0\cdot r+1\leq t\leq s_0\cdot r+s_1-2,\\
        f_5(s_0\cdot r+s_1-1)=s_0\cdot r+s_1-1,\ s\geq 2, \\
        f_5(s_0\cdot r+s_1+d)=(d+1)s_0-1,\ 0\leq d\leq r-1,\\
        f_5(s_0\cdot r+s_1+r)=s_0\cdot r+s_1+r,\ s\geq 2\\
        f_5(t)=m-t-1-(-1)^{t+s_1+r+s_0\cdot r},\ s\geq 6,\ s_0\cdot r+s_1+r+1\leq t\leq s_0\cdot r+s+r-2,\\
        f_5(s_0\cdot r+s+r-1)=s_0\cdot r+\frac{1-(-1)^{s_1}}{2},\ s\geq 4.
\end{array}$$
$$\begin{array}{llll}
        f_6(t)=f_4(t),\ 0\leq t\leq s_0\cdot r-1 \mbox{ or } m-s_0\cdot r\leq t\leq m-1,\hspace{400bp}\\
        f_6(s_0\cdot r)=s_0\cdot r+s+r-1, \ s\geq 1,\\
        f_6(t)=m-t-1,\ s\geq 3, \ s_0\cdot r+1\leq t\leq s_0\cdot r+s_1-1+\frac{1-(-1)^{s}}{2}\\
        f_6(s_0\cdot r+s_1+d+\frac{1-(-1)^{s}}{2})=(d+1)s_0-1,\ 0\leq d\leq r-1, \\
        f_6(t)=m-t-1,\ s\geq 2, \ s_0\cdot r+s_1+r+\frac{1-(-1)^{s}}{2}\leq t\leq s_0\cdot r+s+r-1.
\end{array}$$
$$\begin{array}{llll}
        f_7(t)=f_4(t),\ 0\leq t\leq s_0\cdot r-1 \mbox{ or } m-s_0\cdot r\leq t\leq m-1,\hspace{400bp}\\
        f_7(s_0\cdot r)=s_0\cdot r+s+r-1,\ s\geq 3,\\
        f_7(t)=m-t-1,\ s\geq 5, \ s_0\cdot r+1\leq t\leq s_0\cdot r+s_1-1,\\
        f_7(s_0\cdot r+d+s_1)=(d+1)s_0-1,\ 0\leq d\leq r-1, \\
        f_7(t)=m-t-1, \ s_0\cdot r+s_1+r\leq t\leq s_0\cdot r+s+r-1,\\
        f_7(m-(d+1)s_0)=s_0\cdot r+s_1+d+1,\ 0\leq d\leq r-1, \\
        f_7(t)=m-t-1+(-1)^{t+r+s_0+s_0\cdot d},\ m-d\cdot s_0-s_0+1\leq t\leq m-d\cdot s_0-2,\ 0\leq d\leq r-1, \\
        f_7(m-d\cdot s_0-1)=d\cdot s_0+\frac{1-(-1)^{s_0}}{2},\ 0\leq d\leq r-1.

\end{array}$$
$$\begin{array}{llll}
        f_8(t)=f_4(t),\ 0\leq t\leq s_0\cdot r-1,\hspace{400bp}\\
        f_8(s_0\cdot r)=s_0\cdot r+s+r-1, s\geq 1,\\
        f_8(t)=m-t-1,\ s\geq 3,\ s_0\cdot r+1\leq t\leq s_0\cdot r+s_1-1,\\
        f_8(s_0\cdot r+d+s_1)=(d+1)s_0-1,\ 0\leq d\leq r-1, \\
        f_8(t)=m-t-1,\ s\geq 2, \ s_0\cdot r+s_1+r\leq t\leq s_0\cdot r+s+r-1,\\
        f_8(m-(d+1)s_0)=s_0\cdot r+s_1+d+1,\ 0\leq d\leq r-1, \\
        f_8(t)=m-t-1+(-1)^{t+r+s_0+s_0\cdot d},\ m-d\cdot s_0-s_0+1\leq t\leq m-d\cdot s_0-2,\ 0\leq d\leq r-1, \\
        f_8(m-d\cdot s_0-1)=d\cdot s_0+\frac{1-(-1)^{s_0}}{2},\ 0\leq d\leq r-1.
\end{array}$$
}
\begin{figure}[h]
\centering
\includegraphics[scale=1.0]{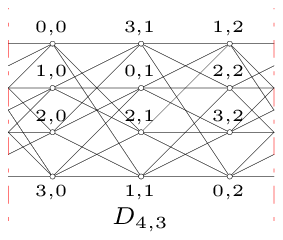}
\includegraphics[scale=1.0]{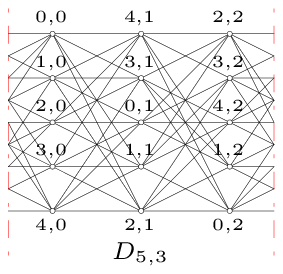}
\includegraphics[scale=1.0]{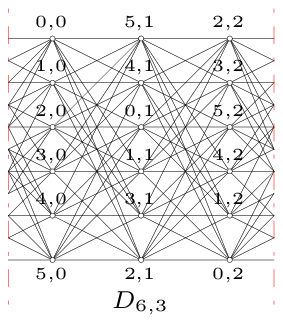}
\includegraphics[scale=1.0]{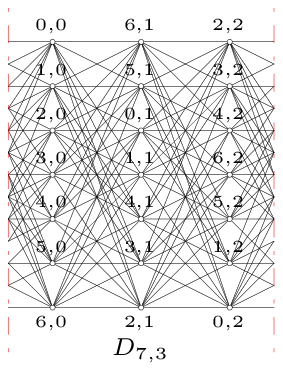}
\caption{\small{Drawings $D_{m,3}$ for $4\leq m\leq 7$}}\label{fig:
D43}
\end{figure}

If $s$ is even, $f^j=f_{5-j \mbox{ mod }2}$ for $0\leq j \leq s-1$,  $f^j=f_6$ for $s\leq j \leq n-1$.\\
\indent If $s$ is odd, $f^j=f_4$ for $0\leq j \leq s-1$,  $f^j=f_{8-j \mbox{ mod }2}$ for $s\leq j \leq n-1$.

For $0\leq i<j \leq m-1$, let
$$\begin{array}{llll}
inv_{l,i,j}&=\left \{ \begin{array}{llll}
1 & \mbox{ if } f_l(i)>f_l(j)\\
0 & \mbox{ if } f_l(i)<f_l(j).
                \end{array}
     \right .
\end{array}$$
Then $inv(f_l)=\sum_{i=0}^{m-1}\sum_{j=i+1}^{m-1}(inv_{l,i,j})$.
\begin{figure}
\centering
\includegraphics[scale=1.0]{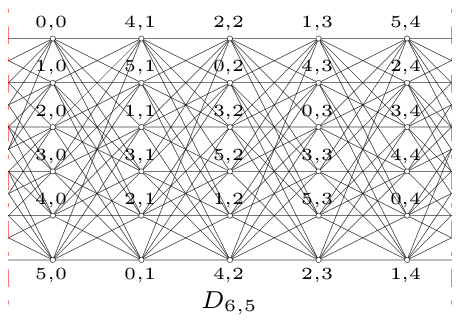}
\includegraphics[scale=1.0]{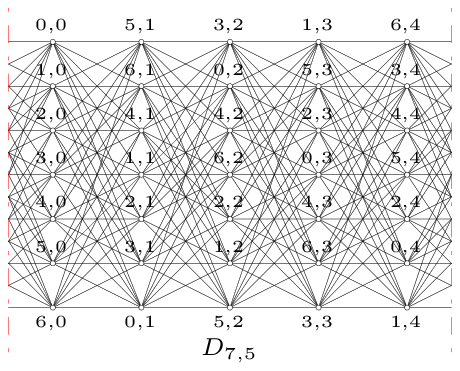}
\includegraphics[scale=1.0]{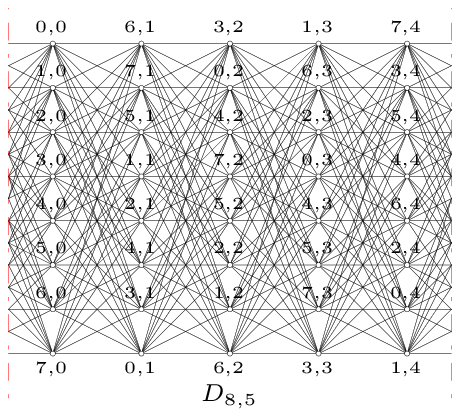}
\includegraphics[scale=1.0]{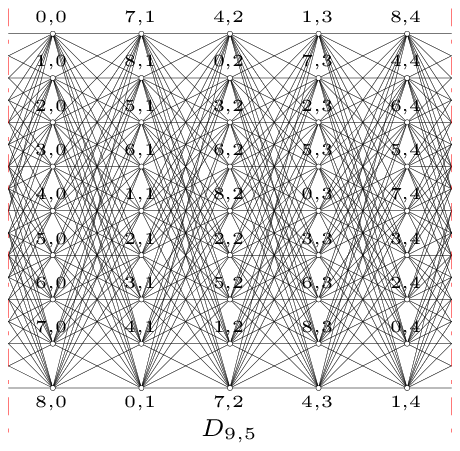}
\includegraphics[scale=1.0]{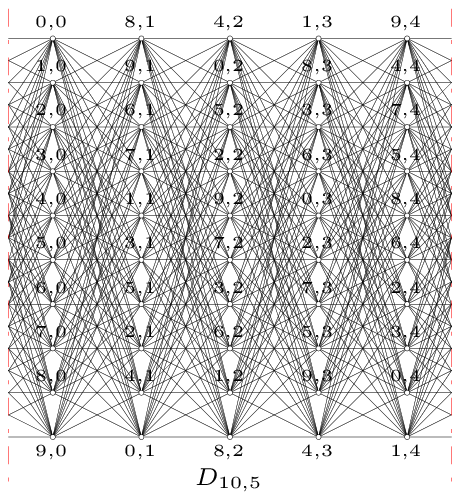}
\includegraphics[scale=1.0]{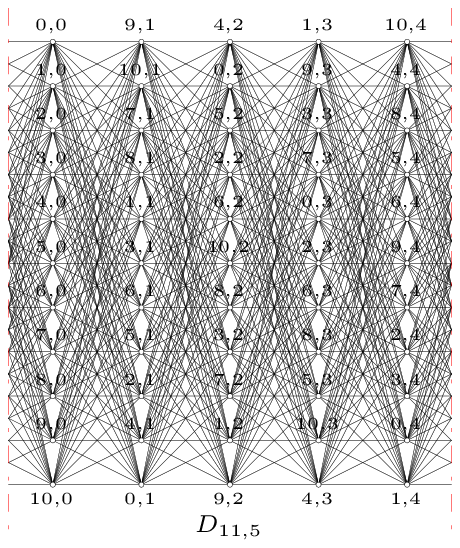}
\caption{\small{Drawings $D_{m,5}$ for $6\leq m\leq 11$}}\label{fig:
D65}
\end{figure}

For $l=4,5,6,7,8$, let
$$\begin{array}{llll}
S_1=&\{i\ |\ 0\leq i\leq s_0\cdot r-1 \}\bigcup_{d=1}^{r}\{i\ |\ m+1-d\cdot s_0\leq m-1-(d-1)d_0\}.
\end{array}$$
For $l=4,5,6,7$, let
$$\begin{array}{llll}
S_2=&\{i\ |\ s_0\cdot r+s_1\leq i\leq s_0\cdot r+s_1+r-1 \}\bigcup_{d=1}^{r}\{m-1-d\cdot s_0+1\},\\
S_3=&\{i\ |\ s_0\cdot r\leq i\leq s_0\cdot r+s_1-1 \}\bigcup \{i\ |\ s_0\cdot r+s_1+r\leq i\leq s_0\cdot r+s\}.
\end{array}$$
for $l=8$, let
$$\begin{array}{llll}
S_2=&\{i\ |\ s_0\cdot r+s_1+1\leq i\leq s_0\cdot r+s_1+r \}\bigcup_{d=1}^{r}\{m-1-d\cdot s_0+1\},\\
S_3=&\{i\ |\ s_0\cdot r\leq i\leq s_0\cdot r+s_1 \}\bigcup \{i\ |\ s_0\cdot r+s_1+r+1\leq i\leq s_0\cdot r+s\}.
\end{array}$$
For $l=4,5,6,7,8$ and $k=1,2,3$, let
$$\begin{array}{llll}
F_{l,k}=\sum_{t\in S_k}(m-1-t)f_l(t)+t(m-1-f_l(t))-\sum_{t\in
S_k}\sum_{j=t+1}^{m-1}(inv_{l,t,j}).
\end{array}$$
By Lemma \ref{lemma: D0}, we have
\begin{lemma}\label{lemma: D2}
If $f^j=f_l$, then
$\nu_{D_{m,n}}(E^j)=(^{m}_2)(^{m}_2)-(F_{l,1}+F_{l,2}+F_{l,3})$.
\end{lemma}

By the definition of $f_l$ and $F_{l,k}$, we have Lemmas \ref{lemma:
D41}-\ref{lemma: D433}:

\begin{lemma}\label{lemma: D41}
For $l=4,5,6,7,8$, {\tiny $$ F_{l,1}=\frac{4(2s_0^{3}-s_0^{2})r^{3}-3((4m+2)s_0^{2}-2m\cdot s_0+1)r^{2}+(-2s_0^{2}+(12m^{2}-12m+10)s_0-6m^{2}-3)r}{6}.$$}
\end{lemma}
\begin{proof} For even $s_0$, we have {\tiny
$$\begin{array}{rlll}
F_{l,1}=&2\sum_{t=1}^{\frac{s_0\cdot r}{2}}((m-2t)(m-2t+1)+(2t-1)(2t-2))\\
 &+\sum_{d=1}^{r}(2\sum_{t=1}^{\frac{s_0-2}{2}}((s_0\cdot d-2t)(s_0\cdot d-2t-1)+(m-s_0\cdot d+2t-1)(m-s_0\cdot d+2t)))\\
 &+\sum_{d=1}^{r}((m-1+s_0-s_0\cdot d)^{2}+(s_0(d-1))^{2})\\
 &-(2\sum_{t=1}^{\frac{s_0\cdot r}{2}}(m-2t)+\sum_{d=1}^{r}(2\sum_{t=1}^{\frac{s_0-2}{2}}(2t-1+(s_0-1)(d-1))+\sum_{d=1}^{r}(s_0-1)(d-1)))\\
=&\frac{4(2s_0^{3}-s_0^{2})r^{3}-3((4m+2)s_0^{2}-2(m+2)s_0+1)r^{2}+(-2s_0^{2}+(12m^{2}-12m+10)s_0-6m^{2}-3)r}{6}.
\end{array}$$
} For odd $s_0$, we have {\tiny
$$\begin{array}{rlll}
F_{l,1}=&\sum_{d=1}^{r}((m-1+s_0-s_0\cdot d)^{2}+(s_0(d-1))^{2})\\
 &+\sum_{d=1}^{r}(2\sum_{t=1}^{\frac{s_0-1}{2}}((m-s_0\cdot d+2t-1)(m-s_0\cdot d+2t-2)+(s_0\cdot d-2t)(s_0\cdot d-2t+1)))\\
 &+\sum_{d=1}^{r}(2\sum_{t=1}^{\frac{s_0-1}{2}}((s_0\cdot d+2t-1)(s_0\cdot d+2t-2)+(m-s_0\cdot d-2t)(m-s_0\cdot d-2t+1)))\\
 &-(\sum_{d=1}^{r}(m-1+s_0-s_0\cdot d)+\sum_{d=1}^{r}(2\sum_{t=1}^{\frac{s_0-1}{2}}(m-1+s_0-s_0\cdot d-2t))+\sum_{d=1}^{r}(2\sum_{t=1}^{\frac{s_0-1}{2}}((s_0-1)(d-1)-2+2t)))\\
=&\frac{4(2s_0^{3}-s_0^{2})r^{3}-3((4m+2)s_0^{2}-2(m+2)s_0+1)r^{2}+(-2s_0^{2}+(12m^{2}-12m+10)s_0-6m^{2}-3)r}{6}.
\end{array}$$
}
\end{proof}
\begin{lemma}\label{lemma: D42}
$$\begin{array}{llll}
F_{l,2}&=\left \{ \begin{array}{llll}
\frac{-(2s_0-1)r^{2}+(-2s_0+2m^{2}-4m+5)r}{2}& \mbox{ if } l=4,5,6\\
\frac{r^{2}+(2m^{2}-6m+3)r}{2}& \mbox{ if } l=7,\\
\frac{(-4s_0+1)r^{2}+(-4s_0+2m^{2}-2m+7)r}{2}& \mbox{ if } l=8.
                \end{array}
     \right .
\end{array}$$
\end{lemma}
\begin{proof} For $l=4,5,6$, we have
{\tiny
$$\begin{array}{rlll}
F_{l,2}=&\sum_{d=1}^{r}((m-s_0\cdot r-s_1-d)(s_0\cdot d-1)+(s_0\cdot r+s_1+d-1)(m-s_0\cdot d))\hspace{400bp}\\
 &+\sum_{d=1}^{r}((s_0\cdot d-1)(s_0\cdot r+s_1+d-1)+(m-s_0\cdot d)(m-s_0\cdot r-s_1-d))\\
 &-(\sum_{d=1}^{r}d(s_0-1)+\sum_{d=1}^{r}(s_0\cdot d-1))\\
=&\frac{-(2s_0-1)r^{2}+(-2s_0+2m^{2}-4m+5)r}{2}.
\end{array}$$
} For $l=7$, we have {\tiny
$$\begin{array}{rlll}
F_{7,2}=&\sum_{d=1}^{r}((m-s_0\cdot r-s_1-d)(s_0\cdot d-1)+(s_0\cdot r+s_1+d-1)(m-s_0\cdot d))\hspace{400bp}\\
 &+\sum_{d=1}^{r}((s_0\cdot d-1)(s_0\cdot r+s_1+d)+(m-s_0\cdot d)(m-s_0\cdot r-s_1-d-1))\\
 &-(\sum_{d=1}^{r}d(s_0-1)+\sum_{d=1}^{r}(s_0\cdot d-1))\\
 =&\frac{r^{2}+(2m^{2}-6m+3)r}{2}.
\end{array}$$
} For $l=8$, we have {\tiny
$$\begin{array}{rlll}
F_{8,2}=&\sum_{d=1}^{r}((m-1-s_0\cdot r-s_1-d)(s_0\cdot d-1)+(s_0\cdot r+s_1+d)(m-s_0\cdot d))\hspace{400bp}\\
 &+\sum_{d=1}^{r}((s_0\cdot d-1)(s_0\cdot r+s_1+d-1)+(m-s_0\cdot d)(m-s_0\cdot r-s_1-d))\\
 &-(\sum_{d=1}^{r}(s_0-1)d+\sum_{d=1}^{r}(s_0\cdot d-1))\\
=&\frac{(-4s_0+1)r^{2}+(-4s_0+2m^{2}-2m+7)r}{2}.
\end{array}$$
}
\end{proof}
\begin{lemma}\label{lemma: D433}
{\tiny
$$\begin{array}{llll}
F_{l,3}=\left \{ \begin{array}{llll}
\frac{12s_0^{2}s_1\cdot r^{2}+3(4s_0\cdot s_1^{2}-(4m\cdot s_0-1)s_1)r+4s_1^{3}-6m\cdot s_1^{2}+(6m^{2}-9m+2)s_1+3\frac{1-(-1)^{s_1}}{2}}{3}& \mbox{ if } l=4 \mbox{ and } s\geq 2  \mbox{ is even},\\
\frac{12(s_0^{2}s_1-2s_0^{2})r^{2}+3(4s_0s_1^{2}-(4(m+4)s_0-1)s_1+(8m+10)s_0+1)r}{3}& \\
+\frac{4s_1^{3}-6(m+4)s_1^{2}+(6m^{2}+15m+32)s_1-(6m^{2}+15m)+3\frac{1+(-1)^{s_1}}{2}}{3}& \mbox{ if } l=5 \mbox{ and } s\geq 2  \mbox{ is even},\\
\frac{12s_0^{2}s_1\cdot r^{2}+3(4s_0\cdot s_1^{2}-(4m\cdot s_0-1)s_1)r+4s_1^{3}-6m\cdot s_1^{2}+(6m^{2}-9m+5)s_1}{3}& \mbox{ if } l=6 \mbox{ and } s  \mbox{ is even},\\
-2(s_0^{2}+2s_0+1)r^{2}+((2m-3)s_0+2(m-1))r& \mbox{ if } l=4 \mbox{ and } s=1,\\
\frac{6(2s_0^{2}\cdot s_1+s_0^{2}+s_0)r^{2} +3(4s_0\cdot s_1^{2}-(4(m-1)s_0-3)s_1-(2m-1)s_0-m-1)r}{3}& \\
+\frac{4s_1^{3}-6(m-1)s_1^{2}+(6m^{2}-15m+5)s_1+3m^{2}-6m+3}{3}& \mbox{ if } l=4 \mbox{ and } s\geq 3  \mbox{ is odd},\\
\frac{6(2s_0^{2}\cdot s_1+s_0^{2})r^{2}+6(2s_0\cdot s_1^{2}-2(m-1)s_0\cdot s_1-(m-1)s_0)r+4s_1^{3}-6(m-1)s_1^{2}+(6m^{2}-12m+8)s_1+3m^{2}-6m+3}{3}& \mbox{ if } l=7 \mbox{ and } s  \mbox{ is odd},\\
\frac{6(2s_0^{2}\cdot s_1+s_0^{2})r^{2}+3(4s_0\cdot s_1^{2}-(4(m-1)s_0-1)s_1-(2m-3)s_0)r+4s_1^{3}-6(m-1)s_1^{2}+(6m^{2}-15m+14)s_1+3m^{2}-9m+6}{3}& \mbox{ if } l=8 \mbox{ and } s  \mbox{ is
odd}.
                \end{array}
     \right .
\end{array}$$
}
\end{lemma}
\begin{proof} We first consider the cases of even $s$. For $l=4$ and even $s_1\geq 2$, we have {\tiny
$$\begin{array}{rlll}
F_{4,3}=&4\sum_{t=1}^{\frac{s_1}{2}}((m-s_0\cdot r-2t)(m-s_0\cdot r-2t+1)+(s_0\cdot r+2t-1)(s_0\cdot r+2t-2))\hspace{400bp}\\
 &-(2\sum_{t=1}^{\frac{s_1}{2}}(m-s_0\cdot r-2t)+2\sum_{t=1}^{\frac{s_1}{2}}(s_0-1)r-2+2t))\\
=&\frac{12s_0^{2}s_1\cdot r^{2}+3(4s_0\cdot s_1^{2}-(4m\cdot s_0-1)s_1)r+4s_1^{3}-6m\cdot s_1^{2}+(6m^{2}-9m+2)s_1}{3}.\\
\end{array}$$
} For $l=4$ and odd $s_1$, we have {\tiny
$$\begin{array}{rlll}
F_{4,3}=&2((m-1-s_0\cdot r)^{2}+(s_0\cdot r)^{2})\\
 &+4\sum_{t=1}^{\frac{s_1-1}{2}}((m-1-s_0\cdot r-2t)(m-s_0\cdot r-2t)+(s_0\cdot r+2t)(s_0\cdot r-1+2t))\hspace{400bp}\\
 &-(m-1-s_0\cdot r+2\sum_{t=1}^{\frac{s_1-1}{2}}(m-1-s_0\cdot r-2t)+2\sum_{t=1}^{\frac{s_1-1}{2}}((s_0-1)r-1+2t)+s_0r-r)\\
=&\frac{12s_0^{2}s_1\cdot r^{2}+3(4s_0\cdot s_1^{2}-(4m\cdot s_0-1)s_1)r+4s_1^{3}-6m\cdot s_1^{2}+(6m^{2}-9m+2)s_1+3}{3}.\\
\end{array}$$
} For $l=5$ and even $s_1\geq 2$, we have {\tiny
$$\begin{array}{rlll}
F_{5,3}=&2((m-1-s_0\cdot r)^{2}+(s_0\cdot r)^{2})\\
 &+4\sum_{t=1}^{\frac{s_1-2}{2}}((m-1-s_0\cdot r-2t)(m-s_0\cdot r-2t)+(s_0\cdot r+2t)(s_0\cdot r-1+2t))\hspace{400bp} \\
 &+2((m-s_0\cdot r-s_1)(s_0\cdot r+s_1-1)+(s_0\cdot r+s_1-1)(m-s_0\cdot r-s1))\\
 &-(m-1-s_0r+2\sum_{t=1}^{\frac{s_1-2}{2}}(m-1-s_0\cdot r-2t)+2(m-1-s_0\cdot r-s_1-r)+2\sum_{t=1}^{\frac{s_1-2}{2}}(s_0-1)r-1+2t)+(s0-1)r)\\
=&\frac{12(s_0^{2}s_1-2s_0^{2})r^{2}+3(4s_0s_1^{2}-(4(m+4)s_0-1)s_1+(8m+10)s_0+1)r+4s_1^{3}-6(m+4)s_1^{2}+(6m^{2}+15m+32)s_1-(6m^{2}+15m)+3}{3}.\\
\end{array}$$
} For $l=5$ and odd $s_1$, we have {\tiny
$$\begin{array}{rlll}
F_{5,3}=&4\sum_{t=1}^{\frac{s_1-1}{2}}((m-s_0\cdot r-2t)(m-s_0\cdot r-2t+1)+(s_0\cdot r+2t-1)(s_0\cdot r+2t-2))\hspace{400bp}\\
 &+2((m-s_0\cdot r-s_1)(s_0\cdot r+s_1-1)+(s_0\cdot r+s_1-1)(m-s_0\cdot r-s1))\\
 &-(2\sum_{t=1}^{\frac{s_1-1}{2}}(m-s_0r-2t)+2(m-1-s_0\cdot r-s_1-r)+2\sum_{t=1}^{\frac{s_1-1}{2}}((s_0-1)r-2+2t))\\
=&\frac{12(s_0^{2}s_1-2s_0^{2})r^{2}+3(4s_0s_1^{2}-(4(m+4)s_0-1)s_1+(8m+10)s_0+1)r+4s_1^{3}-6(m+4)s_1^{2}+(6m^{2}+15m+32)s_1-(6m^{2}+15m)}{3}.\\
\end{array}$$
} For $l=6$, we have{\tiny
$$\begin{array}{rlll}
F_{6,3}=&2\sum_{t=1}^{s_1}((m-s_0\cdot r-t)(m-s_0\cdot r-t)+(s_0\cdot r+t-1)+(s_0\cdot r+t-1)\hspace{400bp}\\
 &-\sum_{t=1}^{s1}(m-s_0r-t)+\sum_{t=1}^{s_1}(s_0-1)r+t-1)\\
=&\frac{12s_0^{2}s_1\cdot r^{2}+3(4s_0\cdot s_1^{2}-(4m\cdot s_0-1)s_1)r+4s_1^{3}-6m\cdot s_1^{2}+(6m^{2}-9m+5)s_1}{3}.\\
\end{array}$$
} Now, we consider the cases of odd $s$. For $l=4$ and even $s_1=0$
$(s=1)$, we have {\tiny
$$\begin{array}{rlll}
F_{4,3}=&(m-1-s_0\cdot r-r)(s_0\cdot r+r)+(s0\cdot r+r)(m-1-s_0\cdot r-r)-s_0\cdot r \hspace{400bp}\\
=&-2(s_0^{2}+2s_0+1)r^{2}+((2m-3)s_0+2(m-1))r.\\
\end{array}$$
} For $l=4$ and even $s_1\geq 2$, we have {\tiny
$$\begin{array}{rlll}
F_{4,3}=&2\sum_{t=1}^{\frac{s_1}{2}}((m-s_0\cdot r-2t)(m-s_0\cdot r-2t+1)+(s_0\cdot r+2t-1)(s_0\cdot r+2t-2))\hspace{400bp}\\
 &+(s_0\cdot r+s_1)(s_0\cdot r+s_1-1)+(m-1-s_0\cdot r-s_1)(m-s_0\cdot r-s_1)\\
 &+(s_0\cdot r+s_1-1)(s_0\cdot r+s_1+r)+(m-s_0\cdot r-s_1)(m-1-s_0\cdot r-s_1-r)\\
 &+2\sum_{t=1}^{\frac{s_1-2}{2}}((s_0\cdot r+s_1-2t)(s_0\cdot r+s_1-2t-1)+(m-1-s_0\cdot r-s_1+2t)(m-s_0\cdot r-s_1+2t))\hspace{400bp}\\
 &+(s_0\cdot r)^{2}+(m-1-s_0\cdot r)^{2}\\
 &-(2\sum_{t=1}^{\frac{s_1}{2}}(m-s_0\cdot r-2t)+(s_0-1)r+s_1-1+s_0\cdot r+s_1-1+2\sum_{t=1}^{\frac{s_1-2}{2}}((s_0-1)r+2t-1)+(s_0-1)r)\\
=&\frac{6(2s_0^{2}\cdot s_1+s_0^{2}+s_0)r^{2} +3(4s_0\cdot s_1^{2}-(4(m-1)s_0-3)s_1-(2m-1)s_0-m-1)r+4s_1^{3}-6(m-1)s_1^{2}+(6m^{2}-15m+5)s_1+3m^{2}-6m+3}{3}.\\
\end{array}$$
} For $l=4$ and odd $s_1$, we have {\tiny
$$\begin{array}{rlll}
F_{4,3}=&(m-1-s_0\cdot r)^{2}+(s_0\cdot r)^{2}\\
 &+2\sum_{t=1}^{\frac{s_1-1}{2}}((m-1-s_0\cdot r-2t)(m-s_0\cdot r-2t)+(s_0\cdot r+2t)(s_0\cdot r-1+2t))\hspace{400bp}\\
 &+(s_0\cdot r+s_1)(s_0\cdot r+s_1-1)+(m-1-s_0\cdot r-s_1)(m-s_0\cdot r-s_1)\\
 &+(s_0\cdot r+s_1-1)(s_0\cdot r+s_1+r)+(m-s_0\cdot r-s_1)(m-1-s_0\cdot r-s_1-r)\\
 &+2\sum_{t=1}^{\frac{s_1-1}{2}}((s_0\cdot r+s_1-2t)(s_0\cdot r+s_1-2t-1)+(m-1-s_0\cdot r-s_1+2t)(m-s_0\cdot r-s_1+2t))\hspace{400bp}\\
 &-(m-1-s_0\cdot r+2\sum_{t=1}^{\frac{s_1-1}{2}}(m-s_0\cdot r-2t)+(s_0-1)r+s_1-1+s_0\cdot r+s_1-1+2\sum_{t=1}^{\frac{s_1-1}{2}}((s_0-1)r+2t-2))\\
=&\frac{6(2s_0^{2}\cdot s_1+s_0^{2}+s_0)r^{2} +3(4s_0\cdot s_1^{2}-(4(m-1)s_0-3)s_1-(2m-1)s_0-m-1)r+4s_1^{3}-6(m-1)s_1^{2}+(6m^{2}-15m+5)s_1+3m^{2}-6m+3}{3}.\\
\end{array}$$
} For $l=7$, we have{\tiny
$$\begin{array}{rlll}
F_{7,3}=&\sum_{t=1}^{s_1}((m-s_0\cdot r-t)(m-s_0\cdot r-t)+(s_0\cdot r+t-1)+(s_0\cdot r+t-1)\hspace{400bp}\\
 &+\sum_{t=1}^{s_1+1}((m-s_0\cdot r-t)(m-s_0\cdot r-t)+(s_0\cdot r+t-1)+(s_0\cdot r+t-1)\\
 &-\sum_{t=1}^{s1}(m-s_0r-t)+\sum_{t=1}^{s_1+1}(s_0-1)r+t-1)\\
=&\frac{6(2s_0^{2}\cdot s_1+s_0^{2})r^{2}+3(4s_0\cdot s_1^{2}-(4(m-1)s_0-1)s_1-(2m-1)s_0+1)r+4s_1^{3}-6(m-1)s_1^{2}+(6m^{2}-15m+8)s_1+3m^{2}-6m+3}{3}.\\
\end{array}$$
} For $l=8$, we have{\tiny
$$\begin{array}{rlll}
F_{8,3}=&\sum_{t=1}^{s_1+1}((m-s_0\cdot r-t)(m-s_0\cdot r-t)+(s_0\cdot r+t-1)+(s_0\cdot r+t-1)\hspace{400bp}\\
 &+\sum_{t=1}^{s_1}((m-s_0\cdot r-t)(m-s_0\cdot r-t)+(s_0\cdot r+t-1)+(s_0\cdot r+t-1)\\
 &-\sum_{t=1}^{s1+1}(m-s_0r-t)+\sum_{t=1}^{s_1}(s_0-1)r+t-1)\\
=&\frac{6(2s_0^{2}\cdot s_1+s_0^{2})r^{2}+3(4s_0\cdot s_1^{2}-(4(m-1)s_0-1)s_1-(2m-3)s_0)r+4s_1^{3}-6(m-1)s_1^{2}+(6m^{2}-15m+14)s_1+3m^{2}-9m+6}{3}.\\
\end{array}$$
}
\end{proof}
By Lemmas \ref{lemma: D2}-\ref{lemma: D433}, we have
\begin{lemma}\label{lemma: D44} For $m>$ odd $n\geq 3$,{\tiny
$$\begin{array}{llll}
\nu(D_{m,n})&=\left \{ \begin{array}{llll} \frac{m(m-1)(m-2)(3m-5)n+2m^{3}-3m^{2}\cdot n+m\cdot n^{2}+4m\cdot n-n^{2}-7m-n+5+(2m\cdot n+4m+13n+8)(\frac{m-1}{n})^{2}}{12}& \mbox{ if } s=1\\
\frac{m(m-1)(m-2)(3m-5)n+(2n^{3}-4n^{2}-8n)\lfloor\frac{m}{n}\rfloor^{3}-(6m\cdot n^{2}-3n^{3}-6m\cdot n-12m+15n)\lfloor\frac{m}{n}\rfloor^{2}}{12}&\\
              +\frac{(6m^{2}n-6m\cdot n^{2}+n^{3}-6m\cdot n+4n^{2}+24m-13n+3n(1-(-1)^{m-n\lfloor\frac{m}{n}\rfloor}))\lfloor\frac{m}{n}\rfloor+6m^{2}-6m}{12}& \mbox{ if } s\neq 1 .
                \end{array}
     \right .
\end{array}$$
}
\end{lemma}
\begin{proof}For even $s$, we have
{\tiny
$$\begin{array}{rlll}
  &\nu(D_{m,n})\\
 =&s_1((^{m}_2)(^{m}_2)-(F_{4,1}+F_{4,2}+F_{4,3})+(^{m}_2)(^{m}_2)-(F_{5,1}+F_{5,2}+F_{5,3}))+(n-2s_1)((^{m}_2)(^{m}_2)-(F_{6,1}+F_{6,2}+F_{6,3})) \hspace{400bp}\\
 =&\frac{(3n\cdot m^{2}(m-1)^{2}-8n(2s_0^{3}-s_0^{2})r^{3}-6(8n\cdot s_0^{2}\cdot s_1-(4m\cdot n-2n)s_0^{2}+2(m\cdot n+n+4s_1+2)s_0)r^{2}}{12}&\\
 &-\frac{2(24n\cdot s_0\cdot s_1^{2}-6(4(m\cdot n+4s_1)s_0-n)s_1-2n\cdot s_0^{2}+(12n\cdot m^{2}-12(n-4s_1)m+4n+60s_1)s_0-12n\cdot m+12n+6s_1)r}{12}&\\
 &-\frac{2(8n\cdot s_1^{3}-12(m\cdot n+4s_1)s_1^{2}+(12n\cdot m^{2}-(18n-48s_1)m+10n+48s_1)s_1-(12m^{2}+30m-6)s_1)}{12}\\
 =&\frac{m(m-1)(m-2)(3m-5)n+(2n^{3}-4n^{2}-8n)\lfloor\frac{m}{n}\rfloor^{3}-(6m\cdot n^{2}-3n^{3}-6m\cdot n-12m+15n)\lfloor\frac{m}{n}\rfloor^{2}}{12}&\\
  &+\frac{(6m^{2}n-6m\cdot n^{2}+n^{3}-6m\cdot n+4n^{2}+24m-13n+3n(1-(-1)^{m-n\lfloor\frac{m}{n}\rfloor}))\lfloor\frac{m}{n}\rfloor+6m^{2}-6m}{12}.
\end{array}$$
} For $s=1$, we have {\tiny
$$\begin{array}{rlll}
  &\nu(D_{m,n})\\
 =&(2s_1+1)((^{m}_2)(^{m}_2)-(F_{4,1}+F_{4,2}+F_{4,3}))+\frac{n-2s_1-1}{2}((^{m}_2)(^{m}_2)-(F_{7,1}+F_{7,2}+F_{7,3})+(^{m}_2)(^{m}_2)-(F_{8,1}+F_{8,2}+F_{8,3})) \hspace{400bp}\\
 =&\frac{3n\cdot m^{2}(m-1)^{2}-8n(2s_0^{3}-s_0^{2})r^{3}-6(-(4m\cdot n-2n+8)s_0^{2}+2(m\cdot n+n-4)s_0-4)r^{2}}{12}\\
 &-\frac{2(-2n\cdot s_0^{2}+(12m^{2}\cdot n-24m(n-1)+16n-30)s_0-12m(n-1)+15(n-1))r+12m^{2}(n-1)-30m(n-1)+18(n-1)}{12}&\\
 =&\frac{m(m-1)(m-2)(3m-5)n+2m^{3}-3m^{2}\cdot n+m\cdot n^{2}+4m\cdot n-n^{2}-7m-n+5+(2m\cdot n+4m+13n+8)(\frac{m-1}{n})^{2}}{12}&\\
\end{array}$$
} For odd $s$, we have {\tiny
$$\begin{array}{rlll}
  &\nu(D_{m,n})\\
 =&(2s_1+1)((^{m}_2)(^{m}_2)-(F_{4,1}+F_{4,2}+F_{4,3}))+\frac{n-2s_1-1}{2}((^{m}_2)(^{m}_2)-(F_{7,1}+F_{7,2}+F_{7,3})+(^{m}_2)(^{m}_2)-(F_{8,1}+F_{8,2}+F_{8,3})) \hspace{400bp}\\
 =&\frac{(3n\cdot m^{2}(m-1)^{2}-8n(2s_0^{3}-s_0^{2})r^{3}-6(8n\cdot s_0^{2}\cdot s_1-(4m\cdot n+2n+16s_1)s_0^{2}+2n(m+1)s_0)r^{2}}{12}&\\
 &-\frac{2(24n\cdot s_0\cdot s_1^{2}-6(4n(m-1)s_0-n-4s_1-2)s_1-2n\cdot s_0^{2}+(12m^{2}\cdot n-24m\cdot n+16n-12s_1-6)s_0-m(12n+12s_1+6)+15n-18s_1-9)r}{12}&\\
 &-\frac{2(8n\cdot s_1^{3}-12n(m-1)s_1^{2}+(12m^{2}\cdot n-30m\cdot n+22n-24s_1-12)s_1+6m^{2}\cdot n-m(15n-6s_1-3)+9n-6s_1-3)}{12}\\
 =&\frac{m(m-1)(m-2)(3m-5)n+(2n^{3}-4n^{2}-8n)\lfloor\frac{m}{n}\rfloor^{3}-(6m\cdot n^{2}-3n^{3}-6m\cdot n-12m+15n)\lfloor\frac{m}{n}\rfloor^{2}}{12}&\\
  &+\frac{(6m^{2}n-6m\cdot n^{2}+n^{3}-6m\cdot n+4n^{2}+24m-13n+3n(1-(-1)^{m-n\lfloor\frac{m}{n}\rfloor}))\lfloor\frac{m}{n}\rfloor+6m^{2}-6m}{12}.
\end{array}$$
}
\end{proof}
\begin{lemma}\label{lemma: D43} For $m>n=3$, $\nu(D_{m,n})\leq\frac{n\cdot m(m-1)(m-2)(3m-5)}{12}+\frac{28m^{3}-54m^{2}+42m+16}{108}.$
\end{lemma}
\begin{proof} By Lemma \ref{lemma: D44}, for $s=1$, we have
{\tiny
$$\begin{array}{rlll}
  &\nu(D_{m,n})\\
 =&\frac{n\cdot m(m-1)(m-2)(3m-5)+2m^{3}-9m^{2}       +9m          +12m      -9    -7m-3+5+(6m       +4m+39 +8)(\frac{m-1}{3})^{2}}{12}\hspace{400bp}\\
 =&\frac{n\cdot m(m-1)(m-2)(3m-5)}{12}+\frac{28m^{3}-54m^{2}+42m-16}{108}\\
 \leq&\frac{n\cdot m(m-1)(m-2)(3m-5)}{12}+\frac{28m^{3}-54m^{2}+42m+16}{108}.
\end{array}$$
} for $s=0$, we have {\tiny
$$\begin{array}{rlll}
  &\nu(D_{m,n})\\
 =&\frac{n\cdot m(m-1)(m-2)(3m-5)}{12}+\frac{(54-36-24)(\frac{m}{3})^{3}-(54m-81-18m-12m+45)(\frac{m}{3})^{2}+(18m^{2}-54m+27-18m+36+24m-39)\frac{m}{3}+6m^{2}-6m}{12}\\
 =&\frac{n\cdot m(m-1)(m-2)(3m-5)}{12}+\frac{28m^{3}-54m^{2}+18m}{108}\\
 \leq&\frac{n\cdot m(m-1)(m-2)(3m-5)}{12}+\frac{28m^{3}-54m^{2}+42m+16}{108}.
\end{array}$$
}for $s=2$, we have {\tiny
$$\begin{array}{rlll}
  &\nu(D_{m,n})\\
 =&\frac{n\cdot m(m-1)(m-2)(3m-5)}{12}+\frac{(54-36-24)(\frac{m-2}{3})^{3}-(54m-81-18m-12m+45)(\frac{m-2}{3})^{2}+(18m^{2}-54m+27-18m+36+24m-39)\frac{m-2}{3}+6m^{2}-6m}{12}\\
 =&\frac{n\cdot m(m-1)(m-2)(3m-5)}{12}+\frac{28m^{3}-54m^{2}+42m+16}{108}.
\end{array}$$
}
\end{proof}

\begin{lemma}\label{lemma: D45} For $m>$ odd $n\geq 5$, $\nu(D_{m,n})\leq\frac{n\cdot m(m-1)(m-2)(3m-5)}{12}+\frac{m^{3}}{4}.$
\end{lemma}
\begin{proof} By Lemma \ref{lemma: D44}, for $s=1$, we have
{\tiny
$$\begin{array}{rlll}
  &\nu(D_{m,n})\\
 =&\frac{n\cdot m(m-1)(m-2)(3m-5)+2m^{3}-3m^{2}\cdot n+m\cdot n^{2}+4m\cdot n-n^{2}-7m-n+5+(2m\cdot n+4m+13n+8)(\frac{m-1}{n})^{2}}{12}\hspace{400bp}\\
 =&\frac{n\cdot m(m-1)(m-2)(3m-5)+2m^{3}-3m^{2}\cdot n+m\cdot n^{2}+4m\cdot n-n^{2}-7m-n+5+(3m\cdot n-m(n-4)+3\times 5n-2(n-4))(\frac{m-1}{n})^{2}}{12}\\
 \leq &\frac{n\cdot m(m-1)(m-2)(3m-5)+2m^{3}-3m^{2}\cdot n+m\cdot n^{2}+4m\cdot n-n^{2}-7m-n+5+(m+3)(m-1)^{2}}{12}\\
 =&\frac{n\cdot m(m-1)(m-2)(3m-5)+3m^{3}-m\cdot n(m-n)-m^{2}(n-1)-m\cdot n(m-4)-n^{2}-12(m-1)-n-4}{12}\\
 \leq&\frac{n\cdot m(m-1)(m-2)(3m-5)}{12}+\frac{m^{3}}{4}\\
\end{array}$$
} for $s\neq 1$, we have {\tiny
$$\begin{array}{rlll}
  &\nu(D_{m,n})\\
 \leq&\frac{n\cdot m(m-1)(m-2)(3m-5)}{12}+\frac{(2n^{3}-4n^{2}-8n)\lfloor\frac{m}{n}\rfloor^{3}-(6m\cdot n^{2}-3n^{3}-6m\cdot n-12m+15n)\lfloor\frac{m}{n}\rfloor^{2}}{12}\hspace{400bp}\\
  &+\frac{(6m^{2}n-6m\cdot n^{2}+n^{3}-6m\cdot n+4n^{2}+24m-7n)\lfloor\frac{m}{n}\rfloor+6m^{2}-6m}{12}\\
 =&\frac{n\cdot m(m-1)(m-2)(3m-5)}{12}+\frac{(2n^{3}-4n^{2}-8n)(\frac{m-s}{n})^{3}-(6m\cdot n^{2}-3n^{3}-6m\cdot n-12m+15n)(\frac{m-s}{n})^{2}}{12}\hspace{400bp}\\
  &+\frac{(6m^{2}n-6m\cdot n^{2}+n^{3}-6m\cdot n+4n^{2}+24m-7n)\frac{m-s}{n}+6m^{2}-6m}{12}\\
 =&\frac{n\cdot m(m-1)(m-2)(3m-5)}{12}+\frac{2m^{3}-3m^{2}\cdot n+m\cdot n^{2}+4m\cdot n-13m-2s^{3}+3n\cdot s^{2}-n0^{2}\cdot s+6m\cdot s-4n\cdot s+7s}{12}\\
  &+\frac{2m^{3}\cdot n+4m^{3}+9m^{2}\cdot n+6m\cdot n\cdot s-6m\cdot n\cdot s^{2}-12m\cdot s^{2}s-15n\cdot s^{2}+8s^{3}+4n\cdot s^{3}}{12n^{2}}\\
 =&\frac{n\cdot m(m-1)(m-2)(3m-5)}{12}+\frac{2m^{3}-3m^{2}\cdot n+m\cdot n^{2}+4m\cdot n-13m-2s^{3}+3n\cdot s^{2}-n0^{2}\cdot s+6m\cdot s-4n\cdot s+7s}{12}\\
  &+\frac{5m^{3}\cdot n+6m\cdot n\cdot s-2m^{2}\cdot n(m-5)-m^{3}(n-4)-m^{2}\cdot n-2m\cdot ns^{2}-4m\cdot s^{2}-15n\cdot s^{2}-8(m-s)s^{2}-4n(m-s)s^{2})}{12n^{2}}\\
 =&\frac{n\cdot m(m-1)(m-2)(3m-5)}{12}+\frac{2m^{3}-3m^{2}\cdot n+m\cdot n^{2}+4m\cdot n-13m-2s^{3}+3n\cdot s^{2}-n0^{2}\cdot s+6m\cdot s-4n\cdot s+7s+m^{3}+6m}{12}\\
\leq&\frac{n\cdot m(m-1)(m-2)(3m-5)}{12}+\frac{3m^{3}-3m^{2}\cdot n+m\cdot n^{2}+4m\cdot n-7m-2(n-1)^{3}+3n\cdot (n-1)^{2}-n0^{2}(n-1)+6m(n-1)-4n(n-1)+7(n-1)}{12}\\
 =&\frac{n\cdot m(m-1)(m-2)(3m-5)}{12}+\frac{3m^{3}-3m^{2}\cdot n+m\cdot n^{2}+10m\cdot n-3n^{2}-13m+8n-5}{12}\\
 =&\frac{n\cdot m(m-1)(m-2)(3m-5)}{12}+\frac{3m^{3}-m\cdot n(m-n)-2m\cdot n(m-5)-13m-n^{2}-2n(n-4)-5}{12}\\
 \leq&\frac{n\cdot m(m-1)(m-2)(3m-5)}{12}+\frac{m^{3}}{4}\\
\end{array}$$
}
\end{proof}

By Lemmas \ref{lemma: D4}, \ref{lemma: D55}, \ref{lemma: D43} and
\ref{lemma: D45}, we have
\begin{theorem}\label{theorem: kmc4}
For $n\geq 3$,
$$\begin{array}{llll}
cr(K_m\times C_n)&\leq\left \{ \begin{array}{llll}
\frac{n\cdot m(m-1)(m-2)(3m-5)}{12}& \mbox{ for } m\geq 4 \mbox{ and even } n\geq 4,\\
\frac{n\cdot m(m-1)(m-2)(3m-5)}{12}+\frac{m(m-1)}{2}& \mbox{ for } 4\leq m\leq \mbox{ odd } n,\\
\frac{n\cdot m(m-1)(m-2)(3m-5)}{12}+\frac{28m^{3}-54m^{2}+42m+16}{108}& \mbox{ for } m>n=3,\\
\frac{n\cdot m(m-1)(m-2)(3m-5)}{12}+\frac{m^{3}}{4}& \mbox{ for } m> \mbox{ odd } n\geq 5.
       \end{array}
     \right .
\end{array}$$
\end{theorem}

\section{Lower bounds of  $cr(K_{m}\times P_n)$ and $cr(K_{m}\times C_n)$}
\indent \indent We shall introduce the lower bound method proposed by Leighton \cite{L84}. Let $G_1=(V_1,E_1)$ and $G_2=(V_2,E_2)$ be graphs. An embedding of $G_1$ in $G_2$ is a couple of
mapping $(\varphi,\kappa)$ satisfying
$$\varphi: V_1\rightarrow V_2 \mbox{ is an injection}$$
$$\kappa: E_1\rightarrow \{\mbox{set of all paths in $G_2$}\},$$ such that
if $uv\in E_1$ then $\kappa(uv)$ is a path between $\varphi(u)$ and $\varphi(v)$. For any $e\in E_2$ define
$$cg_e(\varphi,\kappa)=|\{f\in E_1:e\in \kappa(f)\}|$$
and
$$cg(\varphi,\kappa)=\max\limits_{e\in E_2}\{cg_e(\varphi,\kappa)\}.$$
The value $cg(\varphi,\kappa)$ is called congestion.

\begin{lemma}\label{Lemma congestion} \cite{L84} Let $(\varphi,\kappa)$ be an embedding of $G_1$ in
$G_2$ with congestion $cg(\varphi,\kappa)$. Let $\Delta(G_2)$ denote the maximal degree of $G_2$. Then
$$cr(G_2)\geq \frac{cr(G_1)}{cg^2(\varphi,\kappa)}-\frac{|V_2|}{2}\Delta^2(G_2).$$
\end{lemma}

Let $K^x_{m,n}$ be the complete bipartite multigraph of $m+n$
vertices, in which every two vertices are joined by $x$ parallel
edges.

According to De Klerk \cite{D07} and Kainen \cite{K72}, the
following lemmas hold.

\begin{lemma}\label{Lemma crossing of Kmn}\cite{D07} $cr(K_{m,n})\geq
0.8594\lfloor\frac{m}{2}\rfloor\lfloor\frac{m-1}{2}\rfloor \lfloor\frac{n}{2}\rfloor\lfloor\frac{n-1}{2}\rfloor $.
\end{lemma}

\begin{lemma}\label{Lemma crossing of xKmn}\cite{K72}
$cr(K^x_{m,n})=x^2cr(K_{m,n})$.
\end{lemma}

Now we are in a position to show the lower bound of $cr(K_{m,m}-mK_2)$ and $cr(K_{m,2m}-mK_{1,2}$.

\begin{theorem}\label{Theorem Lower Bound for Kmm}
$cr(K_{m,m}-mK_2)\geq \frac{0.8594}{(1+\frac{3}{m-1})^2}\lfloor\frac{m}{2}\rfloor ^2\lfloor\frac{m-1}{2}\rfloor ^2-m(m-1)^2$.
\end{theorem}
\begin{proof} By Lemmas \ref{Lemma congestion}-\ref{Lemma crossing of xKmn}, we only need to construct an embedding $(\varphi,\kappa)$ of $K^{(m-1)(m-2)}_{m,m}$ into $K_{m,m}-mK_2$ with congestion
$cg(\varphi,\kappa)=(m-2)(m+2)$.

Let $\alpha^k_{i}\beta^k_{i}$ be the $k$-th $(m-1,2)$-arrangement, where $\alpha^k_{i},\beta^k_{i}\in \{0,1,2,\cdots,m-1\}-\{i\}$ and $\alpha^k_{i}\neq \beta^k_{i}$. Let
$$\begin{array}{llll}
V(K^{(m-1)(m-2)}_{m,m})&=\{u_i,v_i\ |\ 0\leq i\leq m-1\},\\
E(K^{(m-1)(m-2)}_{m,m})&=\{(u_i,v_j)^k\ |\ 0\leq i,j\leq m-1, 1\leq k\leq (m-1)(m-2)\},\\
V(K_{m,m}-mK_2)&=\{a_i,b_i\ |\ 0\leq i\leq m-1\},\\
E(K_{m,m}-mK_2)&=\{(a_i,b_j)\ |\ 0\leq i\neq j\leq m-1\}.
\end{array}$$
Let $\varphi (u_i)=a_i$, $\varphi (v_i)=b_i$,
$\kappa((u_i,v_i)^k)=P_{a_ib_{\alpha^k_{i}}a_{\beta^k_{i}}b_i}$ for
$0\leq i\leq m-1$, and $\kappa((u_i,v_j)^k)=(a_i,b_j)$ for $0\leq
i\neq j\leq m-1$. Then $cg_e(\varphi,\kappa)=(m-2)(m+2)$ for every
$e\in E(K^{(m-1)(m-2)}_{m,m})$. This completes the proof of Theorem
\ref{Theorem Lower Bound for Kmm}.
\end{proof}

\begin{theorem}\label{Theorem Lower Bound for Km2m}
$cr(K_{m,2m}-mK_{1,2})\geq \frac{0.8594}{(1+\frac{3}{m-1})^2}m(m-1)\lfloor\frac{m}{2}\rfloor\lfloor\frac{m-1}{2}\rfloor-6m(m-1)^2$.
\end{theorem}
\begin{proof} By Lemmas \ref{Lemma congestion}-\ref{Lemma crossing of xKmn}, we only need to construct an embedding $(\varphi,\kappa)$ of $K^{(m-1)(m-2)}_{m,2m}$ into $K_{m,2m}-mK_{1,2}$ with
congestion $cg(\varphi,\kappa)=(m-2)(m+2)$.

Let $\alpha^k_{i}\beta^k_{i}$ be the $k$-th $(m-1,2)$-arrangement, where $\alpha^k_{i},\beta^k_{i}\in \{0,1,2,\cdots,m-1\}-\{i\}$ and $\alpha^k_{i}\neq \beta^k_{i}$. Let
$$\begin{array}{llll}
V(K^{(m-1)(m-2)}_{m,2m})&=\{u_i,v_i,w_i\ |\ 0\leq i\leq m-1\},\\
E(K^{(m-1)(m-2)}_{m,2m})&=\{(u_i,v_j)^k,(u_i,w_j)^k\ |\ 0\leq i,j\leq m-1, 1\leq k\leq (m-1)(m-2)\},\\
V(K_{m,2m}-mK_{1,2})&=\{a_i,b_i,c_i\ |\ 0\leq i\leq m-1\},\\
E(K_{m,2m}-mK_{1,2})&=\{(a_i,b_j),(a_i,c_j)\ |\ 0\leq i\neq j\leq
m-1\}.
\end{array}$$
Let $\varphi (u_i)=a_i$, $\varphi (v_i)=b_i$, $\varphi (w_i)=c_i$,
$\kappa((u_i,v_i)^k)=P_{a_ib_{\alpha^k_{i}}a_{\beta^k_{i}}b_i}$,
$\kappa((u_i,w_i)^k)=P_{a_ic_{\alpha^k_{i}}a_{\beta^k_{i}}c_i}$ for
$0\leq i\leq m-1$, and $\kappa((u_i,v_j)^k)=(a_i,b_j)$,
$\kappa((u_i,w_j)^k)=(a_i,c_j)$ for $0\leq i\neq j\leq m-1$. Then
$cg_e(\varphi,\kappa)=(m-2)(m+2)$ for every $e\in
E(K^{(m-1)(m-2)}_{m,2m})$. This completes the proof of Theorem
\ref{Theorem Lower Bound for Km2m}.
\end{proof}

Let $D_P$ $(D_C)$ be an arbitrary drawing of $K_m\times P_n$
$(K_m\times C_n)$.  By Lemma \ref{lemma: E12}, we have $\nu(D_P)\geq
\sum_{j=0}^{n-2}\nu_{D_P}(E_j)$ ($\nu(D_C)\geq
\sum_{j=0}^{n-1}\nu_{D_C}(E_j)$). Since $(K_m\times P_n)[E^j]\cong
(K_m\times C_n)[E^j]\cong K_{m,m}-mK_2$ and $(K_m\times P_n)[E^j\cup
E^{j+1}]\cong (K_m\times C_n)[E^j\cup E^{j+1}]\cong
K_{m,2m}-mK_{1,2}$, where $G[X]$ denotes the subgraph of $G$ induced
by $X\subseteq E(G)$, by Theorems \ref{Theorem Lower Bound for Kmm}
and \ref{Theorem Lower Bound for Km2m}, we have

\begin{theorem}\label{Theorem Lower Bound for KmPn}
$$\begin{array}{llll}cr(K_{m}\times P_n)\geq&\left \{\begin{array}{llll}
\frac{n-2}{2}(\frac{0.8594}{(1+\frac{3}{m-1})^2}m(m-1)\lfloor\frac{m}{2}\rfloor\lfloor\frac{m-1}{2}\rfloor-6m(m-1)^2)&\\
+(\frac{0.8594}{(1+\frac{3}{m-1})^2}\lfloor\frac{m}{2}\rfloor ^2\lfloor\frac{m-1}{2}\rfloor ^2-m(m-1)^2)& \mbox{ for even } n\\
\frac{n-1}{2}(\frac{0.8594}{(1+\frac{3}{m-1})^2}m(m-1)\lfloor\frac{m}{2}\rfloor\lfloor\frac{m-1}{2}\rfloor-6m(m-1)^2)&\mbox{ for odd } n.
                       \end{array}
           \right.
\end{array}$$
\end{theorem}

\begin{theorem}\label{Theorem Lower Bound for KmCn}
$$\begin{array}{llll}cr(K_{m}\times C_n)\geq&\left \{\begin{array}{llll}
\frac{n-1}{2}(\frac{0.8594}{(1+\frac{3}{m-1})^2}m(m-1)\lfloor\frac{m}{2}\rfloor\lfloor\frac{m-1}{2}\rfloor-6m(m-1)^2)&\\
+(\frac{0.8594}{(1+\frac{3}{m-1})^2}\lfloor\frac{m}{2}\rfloor ^2\lfloor\frac{m-1}{2}\rfloor ^2-m(m-1)^2)& \mbox{ for odd } n\\
\frac{n}{2}(\frac{0.8594}{(1+\frac{3}{m-1})^2}m(m-1)\lfloor\frac{m}{2}\rfloor\lfloor\frac{m-1}{2}\rfloor-6m(m-1)^2)&\mbox{ for even } n.
                       \end{array}
           \right.
\end{array}$$
\end{theorem}

\end{document}